\documentclass[10pt,a4paper,english]{article}\usepackage[]{graphicx}\usepackage[]{color}
\makeatletter
\def\maxwidth{ %
  \ifdim\Gin@nat@width>\linewidth
    \linewidth
  \else
    \Gin@nat@width
  \fi
}
\makeatother

\definecolor{fgcolor}{rgb}{0.345, 0.345, 0.345}

\usepackage{framed}
\makeatletter
 {\par\unskip\endMakeFramed%
 \at@end@of@kframe}
\makeatother

\definecolor{shadecolor}{rgb}{.97, .97, .97}
\definecolor{messagecolor}{rgb}{0, 0, 0}
\definecolor{warningcolor}{rgb}{1, 0, 1}
\definecolor{errorcolor}{rgb}{1, 0, 0}
\newenvironment{knitrout}{}{} 

\usepackage{alltt}

\usepackage[hyphens]{url}
\usepackage{amsmath}
\usepackage{amsfonts}
\usepackage{amssymb}
\usepackage{amsthm}

\theoremstyle{plain}
\newtheorem{theorem}{Theorem}[section]

\theoremstyle{definition}

\newtheorem{assumption}[theorem]{Assumption}

\theoremstyle{remark}

\PassOptionsToPackage{hyphens}{url}\usepackage{hyperref}
\usepackage{hyperref}
\usepackage[square,numbers]{natbib}

\usepackage[newitem,newenum,increaseonly]{paralist}


\usepackage[notheorems]{SharpeR}

    \usepackage{color} 
    \usepackage{enumerate} 
    \usepackage{amsmath} 
    \usepackage{amssymb} 
		\usepackage{fancyvrb} 
    \usepackage{hyperref}

    \definecolor{orange}{cmyk}{0,0.4,0.8,0.2}
    \definecolor{darkorange}{rgb}{.71,0.21,0.01}
    \definecolor{darkgreen}{rgb}{.12,.54,.11}
    \definecolor{myteal}{rgb}{.26, .44, .56}
    \definecolor{gray}{gray}{0.45}
    \definecolor{lightgray}{gray}{.95}
    \definecolor{mediumgray}{gray}{.8}
    \definecolor{inputbackground}{rgb}{.95, .95, .85}
    \definecolor{outputbackground}{rgb}{.95, .95, .95}
    \definecolor{traceback}{rgb}{1, .95, .95}
    \definecolor{red}{rgb}{.6,0,0}
    \definecolor{green}{rgb}{0,.65,0}
    \definecolor{brown}{rgb}{0.6,0.6,0}
    \definecolor{blue}{rgb}{0,.145,.698}
    \definecolor{purple}{rgb}{.698,.145,.698}
    \definecolor{cyan}{rgb}{0,.698,.698}
    \definecolor{lightgray}{gray}{0.5}
    
    \definecolor{darkgray}{gray}{0.25}
    \definecolor{lightred}{rgb}{1.0,0.39,0.28}
    \definecolor{lightgreen}{rgb}{0.48,0.99,0.0}
    \definecolor{lightblue}{rgb}{0.53,0.81,0.92}
    \definecolor{lightpurple}{rgb}{0.87,0.63,0.87}
    \definecolor{lightcyan}{rgb}{0.5,1.0,0.83}

    \definecolor{incolor}{rgb}{0.0, 0.0, 0.5}
    \definecolor{outcolor}{rgb}{0.545, 0.0, 0.0}

    \hypersetup{
      breaklinks=true,  
      colorlinks=true,
      urlcolor=blue,
      linkcolor=darkorange,
      citecolor=darkgreen,
      }
    
\providecommand{\pqlUL}[3]{\funcitUL{q}{#1}{#2}{#3}}
\providecommand{\pql}[1]{\pqlUL{}{}{#1}}
\providecommand{\pqlsq}[1]{\pqlUL{2}{}{#1}}

\providecommand{\pQlUL}[3]{\funcitUL{Q}{#1}{#2}{#3}}
\providecommand{\pQl}[1]{\pQlUL{}{}{#1}}

\providecommand{\qlev}[1][{}]{\MATHIT{s}}

\providecommand{\qbnd}[1][{}]{\mathUL{B}{}{#1}}

\providecommand{\prvcsym}{c}
\providecommand{\cfncUL}[3]{\funcitUL{\prvcsym}{#1}{#2}{#3}}
\providecommand{\cfnc}[2][{\ssiz}]{\cfncUL{}{#1}{#2}}
\providecommand{\csqfnc}[2][{\ssiz}]{\cfncUL{2}{#1}{#2}}
\providecommand{\cfncp}[2][{\ssiz}]{\cfncUL{\prime}{#1}{#2}}

\providecommand{\csqfnc}[2][{\ssiz}]{\sfncUL{2}{#1}{#2}}

\providecommand{\prvbterm}[3]{\funcitUL{\vect{b}}{#1}{#2}{#3}}
\providecommand{\buterm}[1][{}]{\prvbterm{#1}{\ssiz}{\pvmu, \pvsig}}
\providecommand{\bterm}{\buterm[{}]}
\providecommand{\bsqterm}{\buterm[\trsym]\buterm}

\providecommand{\prvBterm}[3]{\funcitUL{\vect{B}}{#1}{#2}{#3}}
\providecommand{\Buterm}[1][{}]{\prvBterm{#1}{\ssiz}{\pvmu, \pvsig, \pfacsig}}
\providecommand{\Bterm}{\Buterm[{}]}

\providecommand{\sphere}[1]{\mathUL{\mathcal{S}}{#1}{}}
\providecommand{\spherep}{\sphere{\nlatf - 1}}
\providecommand{\fnorm}[1]{\funcitUL{f}{}{\mathcal{S}}{#1}}
\providecommand{\sportwfnc}[1]{\funcit{\sportw}{#1}}

\providecommand{\sportWfnc}[1]{\funcit{\sportW}{#1}}

\providecommand{\spang}{\MATHIT{\theta}}
\providecommand{\spangfnc}[1]{\funcit{\spang}{#1}}

\providecommand{\prskvecU}[1]{\vectUL{\eta}{#1}{}}
\providecommand{\prskvec}{\prskvecU{}}
\providecommand{\gramprskvec}{\prskvecU{\trsym}\prskvec}
\providecommand{\ogramprskvec}{\prskvec\prskvecU{\trsym}}
\providecommand{\Drv}{\Mtx{D}}

\providecommand{\prskMtxU}[1]{\MtxUL{H}{#1}{}}
\providecommand{\prskMtx}{\prskMtxU{}}
\providecommand{\gramprskMtx}{\prskMtxU{\trsym}\prskMtx}

\providecommand{\zvc}{\vect{z}}
\providecommand{\farcsin}[1]{\funcit{\arcsin}{#1}}

\providecommand{\farctan}[1]{\funcit{\arctan}{#1}}

\providecommand{\txtQual}{\txtSNR}
\providecommand{\txtKS}{Kolmogorov-Smirnov\xspace}

\IfFileExists{upquote.sty}{\usepackage{upquote}}{}
\begin{document}

\title{Bounds on Portfolio Quality}
\author{Steven E. Pav \thanks{\email{spav@alumni.cmu.edu} The author thanks
Ramakrishna Kakarala for sharing his research.}}

\maketitle

\begin{abstract}
The \txtQual of a portfolio of \nlatf assets, its expected return divided 
by its risk, is couched as an estimation problem on the sphere
\spherep.  When the portfolio is built using noisy data, the expected
value of the \txtQual is bounded from above via a \txtCR bound, 
for the case of Gaussian returns. The bound holds
for `biased' estimators, thus there appears to be no bias-variance
tradeoff for the problem of maximizing the \txtQual. 
An approximate distribution of the \txtQual for the \txtMP is 
given, and shown to be fairly accurate via Monte Carlo simulations, 
for Gaussian returns as well as more exotic returns distributions.
These findings imply that if the maximal population 
\txtSNR grows slower than the universe size to the 
\oneby{4} power, there may be no diversification benefit, 
rather expected \txtQual can \emph{decrease} with additional assets.
As a practical matter, this may explain why the
\txtMP is typically applied to small asset universes.
Finally, the theorem is expanded to cover more general models
of returns and trading schemes, including the conditional expectation
case where mean returns are linear in some observable
features, subspace constraints (\ie dimensionality reduction),
and hedging constraints.
\end{abstract}

\section{Introduction}

Given \nlatf assets with expected return \pvmu and covariance of return \pvsig,
the portfolio defined as 
\begin{equation}
\pportwopt \defeq \minvAB{\pvsig}{\pvmu},
\end{equation}
known, somewhat informally, as the `\txtMP',
plays a central role in portfolio 
theory. \cite{markowitz1952portfolio,brandt2009portfolio}
Up to scaling, it solves the classic mean-variance
optimization, as well as the 
(population) \txtSR maximization problem:
\begin{equation}
\max_{\pportw }
\frac{\trAB{\pportw}{\pvmu}}{\sqrt{\qform{\pvsig}{\pportw}}}.
\label{eqn:opt_port_I}
\end{equation}

In practice, the \txtMP has a tarnished reputation, and
is infrequently, if ever, used without some modification.
The unknown population parameters \pvmu and \pvsig must
be estimated from samples, resulting in a feasible
portfolio of dubious value. Michaud went so far as to call
mean-variance optimization, ``error maximization.'' \cite{michaud1989markowitz}  
In its stead, numerous portfolio construction methodologies have been proposed
to replace the \txtMP, 
some based on patching conjectured theoretical deficiencies, 
others relying on simple heuristics. \cite{demiguel2009optimal,tu2011markowitz,brandt2009portfolio}

Praticioners often resort to dimensionality reduction heuristics
to mitigate estimation error, effectively reducing the number 
of free variables in the portfolio optimization problem. 
One version of this tactic describes the returns of
dozens, or even hundreds, of equities as the linear combination
of a handful of `factor' returns (plus some `idiosyncratic' term); 
the portfolio problem is then couched as an optimization over factor
portfolios. If the population parameters were known with certainty,
shrinking the set of feasible portfolios would only result in
reducing the optimal portfolio utility. However, the population
parameters can typically only be weakly estimated, and dimensionality
reduction is common practice. 

In this paper, an upper bound is established on the expected value of
a feasible portfolio's \txtQual, defined to be the expected return of the portfolio
divided by it's risk, with return and risk measured using the (unknown)
population parameters, and with the ``expected value'' taken over realizations
of the sample used to estimate the portfolio. This bound balances the
`effect size,' $\sqrt{\ssiz \qiform{\pvsig}{\pvmu}},$ with the number of assets,
\nlatf, and justifies some form of dimensionality reduction. It
is established, for example, that if, by adding additional assets
to the investment universe, $\sqrt{\qiform{\pvsig}{\pvmu}}$ grows at a rate
slower than $\nlatf^{1/4}$, the upper bound on expected \txtQual
can decrease.

\section{Portfolio \txtQual}
\label{sec:portfolio_qual}

Let \vreti be the vector of relative returns of \nlatf
assets, with expectation \pvmu and covariance \pvsig.
A portfolio \sportw on these assets has expected return
\trAB{\sportw}{\pvmu} and variance \qform{\pvsig}{\sportw}. 
Define the \txtQual of the portfolio \sportw as the
\txtQual of the returns of \trAB{\sportw}{\vreti}:
\begin{equation}
\pql{\sportw} \defeq
\frac{\trAB{\sportw}{\pvmu}}{\sqrt{\qform{\pvsig}{\sportw}}}
\label{eqn:def_pql}
\end{equation}

One can think of the \txtQual as a kind of `quality' metric
on portfolios, as follows:
The \txtSR statistic of the future returns of \sportw are 
`stochastically monotonic' in the \txtQual as so defined, 
meaning that if $\pql{\sportw[1]} \le \pql{\sportw[2]}$ then the 
\txtSR of \trAB{\sportw[2]}{\vreti} 
(first order) stochastically dominates the 
\txtSR of \trAB{\sportw[1]}{\vreti}. 

Note that the portfolio \txtQual is bounded by the \txtQual achieved
by the population \txtMP, \pportwopt:
\begin{equation}
\abs{\pql{\sportw}} 
\le \psnropt \defeq \sqrt{\qiform{\pvsig}{\pvmu}} 
= \pql{\pportwopt} 
= \pql{\minvAB{\pvsig}{\pvmu}}.
\end{equation}

We can interpret portfolio \txtQual geometrically, in `risk space',
by introducing a risk transform:
\begin{equation}
\pql{\sportw} =
\frac{\trAB{\sportw}{\pvsig\minvAB{\pvsig}{\pvmu}}}{\sqrt{\qform{\pvsig}{\sportw}}}
=
\frac{\trAB{\wrapParens{\trchol{\pvsig}\sportw}}{\trchol{\pvsig}{\pportwopt}}}{\sqrt{\gram{\wrapParens{\trchol{\pvsig}\sportw}}}}.
\end{equation}
Now normalize by the maximum absolute value that \pql{\sportw} can take:
\begin{equation*}
\begin{split}
\frac{\pql{\sportw}}{\psnropt} 
&=
\frac{\trAB{\wrapParens{\trchol{\pvsig}\sportw}}{\trchol{\pvsig}{\pportwopt}}}{%
\sqrt{\gram{\wrapParens{\trchol{\pvsig}\sportw}}}
\sqrt{\gram{\wrapParens{\trchol{\pvsig}\pportwopt}}}},\\
&=
\trAB{\wrapParens{\frac{\trchol{\pvsig}\sportw}{\sqrt{\gram{\wrapParens{\trchol{\pvsig}\sportw}}}}}}{%
\wrapParens{\frac{\trchol{\pvsig}\pportwopt}{\sqrt{\gram{\wrapParens{\trchol{\pvsig}\pportwopt}}}}}},\\
&=
\trAB{\fnorm{\trchol{\pvsig}\sportw}}{\fnorm{\trchol{\pvsig}\pportwopt}},
\end{split}
\end{equation*}
where 
\begin{equation}
\fnorm{\vect{x}}\defeq \frac{\vect{x}}{\sqrt{\gram{\vect{x}}}}
\end{equation}
is the projection operator taking non-zero vector \vect{x} to the
unit sphere.
That is, \fracc{\pql{\sportw}}{\psnropt} can be viewed as the dot product
of two vectors on the unit sphere (assuming both \sportw and \pportwopt
are non-zero vectors), namely 
\fnorm{\trchol{\pvsig}\sportw} and \fnorm{\trchol{\pvsig}\pportwopt}.
Let \spang be the angle between 
\fnorm{\trchol{\pvsig}\sportw} and \fnorm{\trchol{\pvsig}\pportwopt}, and 
thus 
$\pql{\sportw} = \psnropt \cos \spang.$


In practice the portfolio \sportw is built using \ssiz \iid observations 
of the random variable \vreti. Denote these observations by the
\bby{\ssiz}{\nlatf} matrix \mreti, and, by abuse of notation, denote
the \emph{estimator} that gives \sportw for a given \mreti by
\sportwfnc{\mreti}. By the same abuse of notation, write 
\spangfnc{\mreti}. We will bound the expected value of
\sportwfnc{\mreti}. 


To appeal to a \txtCR bound, one must typically assume the estimator
is unbiased. For this problem a somewhat weaker condition suffices.
\begin{assumption}[Directional Independence]
Assume that
\begin{equation}
\label{eqn:sane_estimator}
\E{\fnorm{\trchol{\pvsig}\sportwfnc{\mreti}}} = 
\cfnc{\psnrsqopt} \fnorm{\trchol{\pvsig}\pportwopt}
+ \bterm,
\end{equation}
where \bterm is the `bias' term, which is orthogonal to
\fnorm{\trchol{\pvsig}\pportwopt}, and which may be an
arbitrary function
of \pvmu and \pvsig.

\end{assumption}

Note that by orthogonality of \bterm and \fnorm{\trchol{\pvsig}\pportwopt}, 
and linearity of the expectation, 
\begin{equation}
\E{\cos \spangfnc{\mreti}} = 
\E{\frac{\pql{\sportw}}{\psnropt}} =
\E{\trAB{\fnorm{\trchol{\pvsig}\sportwfnc{\mreti}}}{\fnorm{\trchol{\pvsig}\pportwopt}}}
= \cfnc{\psnrsqopt}.
\end{equation}
Thus $\abs{\cfnc{x}} \le 1$,
and we expect $\cfnc{x} \ge 0$ for a `sane' portfolio estimator.
Moreover, one expects $\cfnc{x} \to 0$ as $\ssiz x \to 0$, and
for non-zero $x$, $\cfnc[{\ssiz}]{x} \to 1$ as $\ssiz\to\infty$.

When \bterm is the zero vector, the estimator is a 
`parallel estimator' in Watson's terminology
\cite{ANZS:ANZS253}, or `unbiased' in the sense of Hendricks.
\cite{Hendriks1991245,Dutchmen1992}
Note that \eqnref{sane_estimator} is satisfied for 
any \emph{directionally equivariant} portfolio estimator,
\ie one which, for any orthonormal \Mtx{H}, 
($\gram{\Mtx{H}} = \eye[\nlatf] = \ogram{\Mtx{H}}$), one
has 
\begin{equation*}
\sportwfnc{\mreti\tr{\Mtx{H}}} = \Mtx{H}\sportwfnc{\mreti}.
\end{equation*}
However, one should recognize that not all portfolio estimators
satisfy this assumption. For example, consider an estimator that
never concentrates greater than $\nlatf^{-\half}$ proportion of its
total gross allocation in any one asset; this estimator does not exhibit
Directional Independence, since it can not capitalize when 
$\pportwopt=\psnropt\basev[1]$. Neither does the ``one over $N$ allocation'' 
estimator.  \cite{demiguel2009optimal}  

We must eliminate other `pathological' cases from consideration.
\begin{assumption}[Residual Independence]
Assume that the distribution of the residual
$$
\fnorm{\trchol{\pvsig}\sportwfnc{\mreti}} - \E{\fnorm{\trchol{\pvsig}\sportwfnc{\mreti}}}
$$
is independent of \trchol{\pvsig}\pportwopt.
\end{assumption}
This assumption prevents us from making false assertions about 
\eg the 1/$N$ allocation in the case where 
it happens to nearly equal \pportwopt.  \cite{demiguel2009optimal}  

Let $\vect{y}$ be a \nlatf-variate random variable. Then
\begin{align}
\nonumber\trace{\VAR{\vect{y}}} 
&= \trace{\E{\ogram{\wrapParens{\vect{y} - \E{\vect{y}}}}}},&\\
\nonumber &= \trace{\E{\ogram{\vect{y}}}} - \trace{\ogram{\E{\vect{y}}}},&\\
&= \E{\gram{\vect{y}}} - \gram{\E{\vect{y}}}.&
\end{align}
By \eqnref{sane_estimator}, and using orthogonality of 
\bterm and \fnorm{\trchol{\pvsig}\pportwopt}, we then have
\begin{align}
\nonumber
\trace{\VAR{\fnorm{\trchol{\pvsig}\sportwfnc{\mreti}}}} &= 
1 - \wrapParens{\csqfnc{\psnrsqopt} + \bsqterm}&\\
\label{eqn:var_bounds}
&\le 1 - \csqfnc{\psnrsqopt},&
\end{align}
We will bound the variance of 
$\fnorm{\trchol{\pvsig}\sportwfnc{\mreti}}$
by a \txtCR lower bound, thus establishing an upper bound on
\cfnc{\psnrsqopt}.

Define 
\begin{equation}
\prskvec \defeq \trchol{\pvsig}\pportwopt = \ichol{\pvsig}\pvmu.
\end{equation}
Note that $\gramprskvec = \qiform{\pvsig}{\pvmu} = \psnrsqopt$.
Using the \txtCR lower bound for the left hand side 
of \eqnref{var_bounds},
and then using the definition of \prskvec in the expectation, we have
\cite{tj_moore_ccrb}
\begin{equation}
\label{eqn:crb_one}
\oneby{\ssiz}\trace{\qoform{\iFishI[\prskvec]}{\Drv}}
\le 1 - \csqfnc{\gramprskvec},
\end{equation}
where
\begin{equation}
\Drv\defeq
{\dbyd{\cfnc{\gramprskvec}\frac{\prskvec}{\sqrt{\gramprskvec}}}{\prskvec}}.
\end{equation}
Here we take the derivative to follow the `numerator layout' convention, 
meaning a gradient is a row vector. This derivative takes the form
\begin{equation}
\label{eqn:Drv_form}
\Drv = \frac{\cfncp{\gramprskvec}}{\sqrt{\gramprskvec}}\ogramprskvec + 
\cfnc{\gramprskvec}\wrapParens{\frac{\eye}{\sqrt{\gramprskvec}} -
\frac{\ogramprskvec}{\gramprskvec^{\half[3]}}}.
\end{equation}

To compute the Fisher information, \FishI[\prskvec], we must fix the likelihood
of the returns, \vreti. While the normal distribution is a poor fit for asset
returns \cite{stylized_facts}, it is a convenient distribution to work with.

\begin{assumption}[Normal Returns]
Assume that \vreti are multivariate normally distributed,
$\vreti\sim\normlaw{\pvmu,\pvsig}$.
\end{assumption}

For multivariate normal returns, and conditional on \pvsig, the log likelihood
takes the form
\begin{equation}
\begin{split}
\log\FOOlik{}{\vreti}{\prskvec} &= c_1 - \half
\qiform{\pvsig}{\wrapParens{\vreti - \pvmu}},\\
&= \funcit{c}{\vreti} + \prskvecU{\trsym}\ichol{\pvsig}\vreti - \half
\gramprskvec,
\end{split}
\end{equation}
dropping the `nuisance parameters' from the likelihood function.
The Fisher Information is negative the expectation of the 
Hessian of the log likelihood with respect to \prskvec. 
In this case we have simply
\begin{equation}
\label{eqn:FisherI}
\FishI[\prskvec] =
- \E{\prby[2]{\log\FOOlik{}{\vreti}{\prskvec}}{\px[\prskvec]\px[\prskvecU{\trsym}]}}
= \eye[\nlatf].
\end{equation}
This radically simplifies the exposition, as the \txtCR bound of
\eqnref{crb_one} can now be expressed as 
\begin{equation}
\label{eqn:crb_two}
\oneby{\ssiz}\trace{\ogram{\Drv}}
\le 1 - \csqfnc{\gramprskvec}.
\end{equation}
Using the form of \Drv given in \eqnref{Drv_form}, and noting that the cross
terms are orthogonal, we have
\begin{equation}
\begin{split}
\trace{\ogram{\Drv}}
&=
\trace{\wrapBracks{\cfncp{\gramprskvec}}^2 \ogramprskvec +
\csqfnc{\gramprskvec}\wrapBracks{\frac{\eye}{\gramprskvec} 
- \frac{\ogramprskvec}{\wrapParens{\gramprskvec}^2}}},\\
&= \wrapBracks{\cfncp{\gramprskvec}}^2 \gramprskvec + 
\csqfnc{\gramprskvec}\frac{\nlatf - 1}{\gramprskvec},
\end{split}
\end{equation}
using the fact that $\trace{\ogram{\vect{y}}} = \gram{\vect{y}}$.
With \eqnref{crb_two}, this gives
\begin{equation}
\label{eqn:crb_three}
\wrapBracks{\cfncp{\gramprskvec}}^2 \gramprskvec + 
\csqfnc{\gramprskvec}\frac{\nlatf - 1}{\gramprskvec}
\le \ssiz\wrapParens{1 - \csqfnc{\gramprskvec}}.
\end{equation}
The term $\wrapBracks{\cfncp{\gramprskvec}}^2 \gramprskvec$ is 
non-negative, so we may discard it to get a coarser bound that 
does not involve the derivative of \cfnc{}:
\begin{equation}
\label{eqn:crb_four}
\csqfnc{\gramprskvec}\frac{\nlatf - 1}{\gramprskvec}
\le \ssiz\wrapParens{1 - \csqfnc{\gramprskvec}}.
\end{equation}
This yields
\begin{equation}
\label{eqn:crb_five}
\csqfnc{\gramprskvec} \le \frac{\ssiz\gramprskvec}{\nlatf - 1 +
\ssiz\gramprskvec},
\end{equation}
proving the following theorem.
\begin{theorem}
\label{theorem:qual_bound}
Let \sportwfnc{\mreti} be an estimator based on \ssiz \iid observations of
multivariate Gaussian returns, \mreti, satisfying the assumptions of
directional independence and residual independence. Then
\begin{equation}
\E{\pql{\sportwfnc{\mreti}}} 
\le \frac{\sqrt{\ssiz}\psnrsqopt}{\sqrt{\nlatf - 1 + \ssiz\psnrsqopt}}.
\end{equation}
\end{theorem}

\theoremref{qual_bound} balances the ``degrees of freedom'' of
the estimator, $\nlatf-1$, with one lost because only direction matters,
and the ``observable effect size'', $\ssiz\psnrsqopt$. The effect size
is a unitless quantity. If \psnropt is measured in trading days, then \ssiz should
be the number of trading days; if \psnropt is measured in 
`annualized' terms, then \ssiz should be the number of years.

This bound is fairly harsh. Consider a typical actively managed portfolio.
Generously, we can estimate $\psnropt=1\yrtomhalf$ over
$\nlatf=10$
assets, using $\ssiz=5\yrto{}$ of historical data. Then the expected
value of \pql{\sportwfnc{\mreti}} is bounded by 
$0.6\yrtomhalf$; the event of having a year-over-year loss is
then a ``$0.6$-sigma'' event. 

\theoremref{qual_bound} suggests that for comparing investments,
the magnitude of the \emph{squared} \txtSR is a limiting factor,
rather than the \txtSR itself (assuming it is positive). That
is, under the bound of the theorem, $\psnropt=2\yrtomhalf$
is \emph{four} times as `good' as $\psnropt=1\yrtomhalf$, in the
sense that such an effect size can `balance' four times as many
degrees of freedom.

\section{Approximate distribution of the \txtQual of the \txtMP}
\label{sec:apx_distribution}

Here we establish an approximate distribution of the 
quantity $\fracc{\pql{\sportw}}{\psnropt} = \cos\spang$
for the sample \txtMP, $\sportwopt \defeq \minvAB{\svsig}{\svmu},$
with $\svsig, \svmu$ the usual sample estimates of \pvsig and
\pvmu. The approximation is constructed by assuming that
misestimation of \pvsig contributes no error to the portfolio.

Assuming that $\svsig = \pvsig$, then
\begin{equation}
\trchol{\pvsig}\sportwopt = \ichol{\pvsig}\svmu = 
\ichol{\pvsig}\pvmu + \oneby{\sqrt{\ssiz}}\zvc,
\end{equation}
where $\zvc \sim \normlaw{\vzero, \eye}$. 
Then, with $\fracc{\pql{\sportwfnc{\mreti}}}{\psnropt} =
\fcos{\spangfnc{\mreti}}$, we should have
\begin{equation}
\fcot{\spangfnc{\mreti}} =  \frac{\norm{\ichol{\pvsig}\pvmu} + \oneby{\sqrt{\ssiz}}z_1}{%
\sqrt{\oneby{\ssiz}\sum_{2 \le i \le \nlatf} z_i^2}},
\end{equation}
where the $z_i$ are independent standard normal random variables.
This can be expressed as
\begin{equation}
\label{apx:qual_dist}
\ftan{\farcsin{\frac{\pql{\sportwfnc{\mreti}}}{\psnropt}}}
\sim \oneby{\sqrt{\nlatf - 1}}\nctlaw{\sqrt{\ssiz}\psnropt,\nlatf-1},
\end{equation}
where \nctlaw{\nctp,\nu} is a non-central \tstat-distribution with 
non-centrality parameter $\nctp$ and $\nu$ degrees of freedom.

\apxref{qual_dist} implies the following approximation:
\begin{equation}
\label{apx:qual_beta}
\pqlsq{\sportwfnc{\mreti}} \sim 
\psnrsqopt \nctbetalaw{\ssiz\psnrsqopt}{\half}{\half[\nlatf-1]},
\end{equation}
where \nctbetalaw{\nctp}{p}{q} is a non-central Beta distribution 
with non-centrality \nctp, and `shape' parameters $p$ and $q$.
\cite{walck:1996}
However, by describing the distribution of the 
\emph{square} of \pql{\sportwfnc{\mreti}}, we cannot easily model the 
(sometimes significant) probability that it 
is negative.
This form, does, however, give bounds on the variance
of \pql{\sportwfnc{\mreti}} under the 
approximation of \apxref{qual_dist}, since
the moments of the non-central Beta are known.
\cite[sec 30.3]{walck:1996} Under \apxref{qual_dist},
we have
\begin{equation}
\label{eqn:hypergeo_extwo}
\E{\pqlsq{\sportwfnc{\mreti}}} = 
\psnrsqopt \exp{-\half[\ssiz\psnrsqopt]}
\GAMhalfrat{3}{1}
\GAMhalfrat{\nlatf}{\nlatf+2}
\HyperF{2}{2}{\half[\nlatf], \half[3] ; \half, \half[2 + \nlatf];
\half[\ssiz\psnrsqopt]},
\end{equation}
where \HyperF{2}{2}{\cdot,\cdot;\cdot,\cdot;\cdot} is the Generalized
Hypergeometric function. \cite[sec 16.2]{NIST_handbook} This
is a rough upper bound on the variance of \pqlsq{\sportwfnc{\mreti}};
a lower bound can be had using the upper bound on the mean
from \theoremref{qual_bound}.

Because the median value of the non-central 
\tstat{}-distribution is approximately equal to the 
non-centrality parameter, \cite{Johnson:1940,kramer_paik_1979}
the median value of \pql{\sportwfnc{\mreti}} for
the sample \txtMP, via \apxref{qual_dist}, is approximately
\begin{equation}
\label{apx:hcut50_apx}
m \approx \psnropt \fsin{\farctan{\frac{\sqrt{\ssiz}\psnropt}{\sqrt{\nlatf-1}}}}
= \frac{\sqrt{\ssiz}\psnrsqopt}{\sqrt{\nlatf - 1 + \ssiz \psnrsqopt}},
\end{equation}
which is exactly the upper bound of \theoremref{qual_bound}!


\subsection{Monte Carlo simulations}

The accuracy of \apxref{qual_dist} is checked by 
Monte Carlo simulations: $\ensuremath{10^{7}}$ simulations were performed
of construction of the \txtMP 
using $\ssiz=1012$ (4 years of daily observations), 
$\nlatf=6$ and $\psnropt=1.25\yrtomhalf$;
the returns are normally distributed.
Since the population \txtMP is known, the portfolio \txtQual 
can be computed exactly.
The Q-Q plot in \figref{haircutting} confirms that 
\apxref{qual_dist} is very good for this choice of $\ssiz, \nlatf, \psnropt$.

Rather than rely on `proof by graph', the \txtKS test was
computed for the values of \txtQual generated under Gaussian returns.
\cite{Marsaglia:Tsang:Wang:2003:JSSOBK:v08i18}
The statistic, the maximal difference between empirical CDF and
theoretical CDF under the approximation, was computed to be 0.004
over the $\ensuremath{10^{7}}$ simulations.
While this seems small, the computed p-value under the null underflows to 0
because the sample size is so large.

\begin{knitrout}\small
\definecolor{shadecolor}{rgb}{0.969, 0.969, 0.969}\color{fgcolor}\begin{figure}[]

\includegraphics[width=\maxwidth]{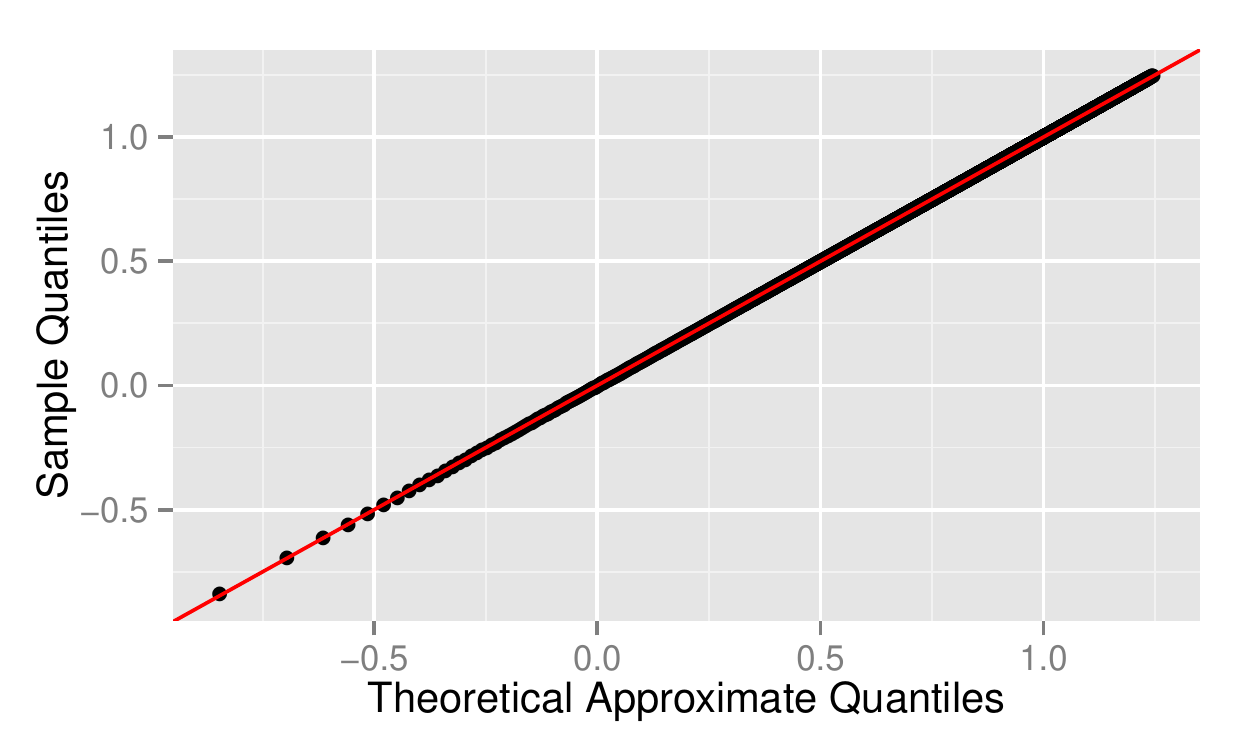} \caption[Q-Q plot of $10^{7}$ simulated \txtQual values versus \apxref{qual_dist} is shown]{Q-Q plot of $10^{7}$ simulated \txtQual values versus \apxref{qual_dist} is shown. Units are `annual', \ie \yrtomhalf. Since the number of samples is very large, only a subset of $10^{4}$ points, uniformly selected by sample quantile, are plotted.\label{fig:haircutting}}
\end{figure}

\end{knitrout}



The experiment is then repeated using returns drawn 
from a uniform distribution, 
a \tstat{}-distribution with $4$ degrees of freedom, 
from a Tukey $h$-distribution with parameter $h=0.15$, 
and from a Lambert W $\times$ Gaussian
with parameter $\gamma=\ensuremath{-0.2}$.  \cite{2009arXiv0912.4554G,2010arXiv1010.2265G}
Returns are generated by first generating \iid \nlatf-variate draws
from a zero mean, identity covariance distribution whose marginals
follow the so-named laws, then scaling and shifting to have the
appropriate \psnropt. For each simulation, the \pvsig is a random
draw from a Wishart random variable.

The uniform distribution is not a realistic model of market returns,
but is included to check the approximation on platykurtic returns.
The \tstat{} and more exotic distributions are more realistic
models of market returns, and are leptokurtotic. The 
Lambert W has non-zero skew. 
Again, $\ensuremath{10^{7}}$ simulations are performed under each of these 
distributions with the same values of $\ssiz, \nlatf, \psnropt$ as
above. 
Some of the empirical quantiles from these simulations are
shown in \tabref{qtiles},
along with the approximate quantiles from \apxref{qual_dist}. 
The \txtKS test statistics for the different
distributions are presented in \tabref{ks_stats}.
For this choice of \ssiz, \nlatf, \psnropt, the approximation is
very good, across the tested returns distributions.

\begin{table}[ht]
\centering
\begin{tabular}{r|ccccc|c}
  \hline
q.tile & normal & unif. & t(4) & Tukey(0.15) & Lam.W(-0.2) & approx. \\ 
  \hline
0.005 & -0.0499 & -0.0493 & -0.0435 & -0.0474 & -0.0514 & -0.0450 \\ 
  0.010 & 0.0947 & 0.0954 & 0.1001 & 0.0973 & 0.0934 & 0.0996 \\ 
  0.025 & 0.2885 & 0.2888 & 0.2923 & 0.2901 & 0.2871 & 0.2928 \\ 
  0.050 & 0.4356 & 0.4356 & 0.4383 & 0.4369 & 0.4341 & 0.4397 \\ 
  0.250 & 0.7858 & 0.7859 & 0.7867 & 0.7864 & 0.7850 & 0.7890 \\ 
  0.500 & 0.9528 & 0.9527 & 0.9531 & 0.9529 & 0.9521 & 0.9550 \\ 
  0.750 & 1.0706 & 1.0706 & 1.0708 & 1.0706 & 1.0702 & 1.0721 \\ 
  0.900 & 1.1432 & 1.1431 & 1.1433 & 1.1433 & 1.1431 & 1.1442 \\ 
   \hline
\end{tabular}
\caption{Empirical quantiles of portfolio \txtQual from $10^{7}$ simulations of 1012 days of 6 assets, with maximal \txtSR of $1.25\yrtomhalf$ are given, along with the approximate quantiles from \apxref{qual_dist}. Units of \txtQual are `annual', \ie \yrtomhalf.} 
\label{tab:qtiles}
\end{table}

\begin{table}[ht]
\centering
\begin{tabular}{ccccc}
  \hline
normal & unif. & t(4) & Tukey(0.15) & Lam.W(-0.2) \\ 
  \hline
0.0045 & 0.0045 & 0.0039 & 0.0043 & 0.0057 \\ 
   \hline
\end{tabular}
\caption{\txtKS statistic comparing the empirical CDF to that of \apxref{qual_dist} over $10^{7}$ simulations of 1012 days of 6 assets, with maximal \txtSR of $1.25\yrtomhalf$ are given for the different returns distributions.} 
\label{tab:ks_stats}
\end{table}


\begin{table}[ht]
\centering
\begin{tabular}{ccccc|c}
  \hline
normal & unif. & t(4) & Tukey(0.15) & Lam.W(-0.2) & bound \\ 
  \hline
0.9 & 0.898 & 0.899 & 0.899 & 0.898 & 0.932 \\ 
   \hline
\end{tabular}
\caption{Empirical mean portfolio \txtQual from $10^{7}$ simulations of 1012 days of 6 assets, with maximal \txtSR of $1.25\yrtomhalf$ are given, along with the upper bound from \theoremref{qual_bound}. Units of \txtQual are `annual', \ie \yrtomhalf.} 
\label{tab:means}
\end{table}
\begin{table}[ht]
\centering
\begin{tabular}{ccccc|c}
  \hline
normal & unif. & t(4) & Tukey(0.15) & Lam.W(-0.2) & approx. \\ 
  \hline
0.9 & 0.864 & 0.865 & 0.864 & 0.863 & 0.868 \\ 
   \hline
\end{tabular}
\caption{Empirical mean of \emph{squared} portfolio \txtQual from $10^{7}$ simulations of 1012 days of 6 assets, with maximal \txtSR of $1.25\yrtomhalf$ are given, along with the approximate value from \eqnref{hypergeo_extwo}. Units of squared \txtQual are `annual', \ie \yrto{-1}.} 
\label{tab:meanssq}
\end{table}

In \tabref{means}, the empirical mean value of \pql{\sportwopt}, over
the $\ensuremath{10^{7}}$ simulations, is presented for the five returns
distributions, along with the upper bound given by
\theoremref{qual_bound}.  It seems that there is a small gap between
the empirical mean for the case of Gaussian returns, and the
theoretical upper bound, a gap on the order of 
$4\%$.
Perhaps this gap is caused by discarding
the derivative term from \eqnref{crb_three}, or because 
the sample Markowitz portfolio is not efficient for finite samples.

In \tabref{meanssq}, the empirical mean value of \pqlsq{\sportwopt}, over
the $\ensuremath{10^{7}}$ simulations, is presented for the five returns
distributions, along with the theoretical value from 
\eqnref{hypergeo_extwo}, which is valid only under 
\apxref{qual_dist}. The approximate value is decent, meaning
an estimate of the variance of \pql{\sportwopt} could be had
by combining \eqnref{hypergeo_extwo} and the upper bound of
\theoremref{qual_bound}.




Of course, these simulations are conducted using only a single
choice of the parameters \ssiz, \nlatf and \psnropt. To check the
robustness of this approximation to these parameters, 
$\ensuremath{10^{5}}$ Monte Carlo simulations were conducted for 
each combination of $\ssiz=0.5, 1, 2, 4, 8$ years of daily observations,
$\nlatf=2, 4, 8, 16, 32$, and $\psnropt=0.35, 0.5, 0.71, 1, 1.41\yrtomhalf$,
all under Gaussian returns. The \txtKS test statistic
is then computed on the empirically observed quantiles of
portfolio \txtQual, under the distribution of
\apxref{qual_dist}.  

\begin{knitrout}\small
\definecolor{shadecolor}{rgb}{0.969, 0.969, 0.969}\color{fgcolor}\begin{figure}[]

\includegraphics[width=\maxwidth]{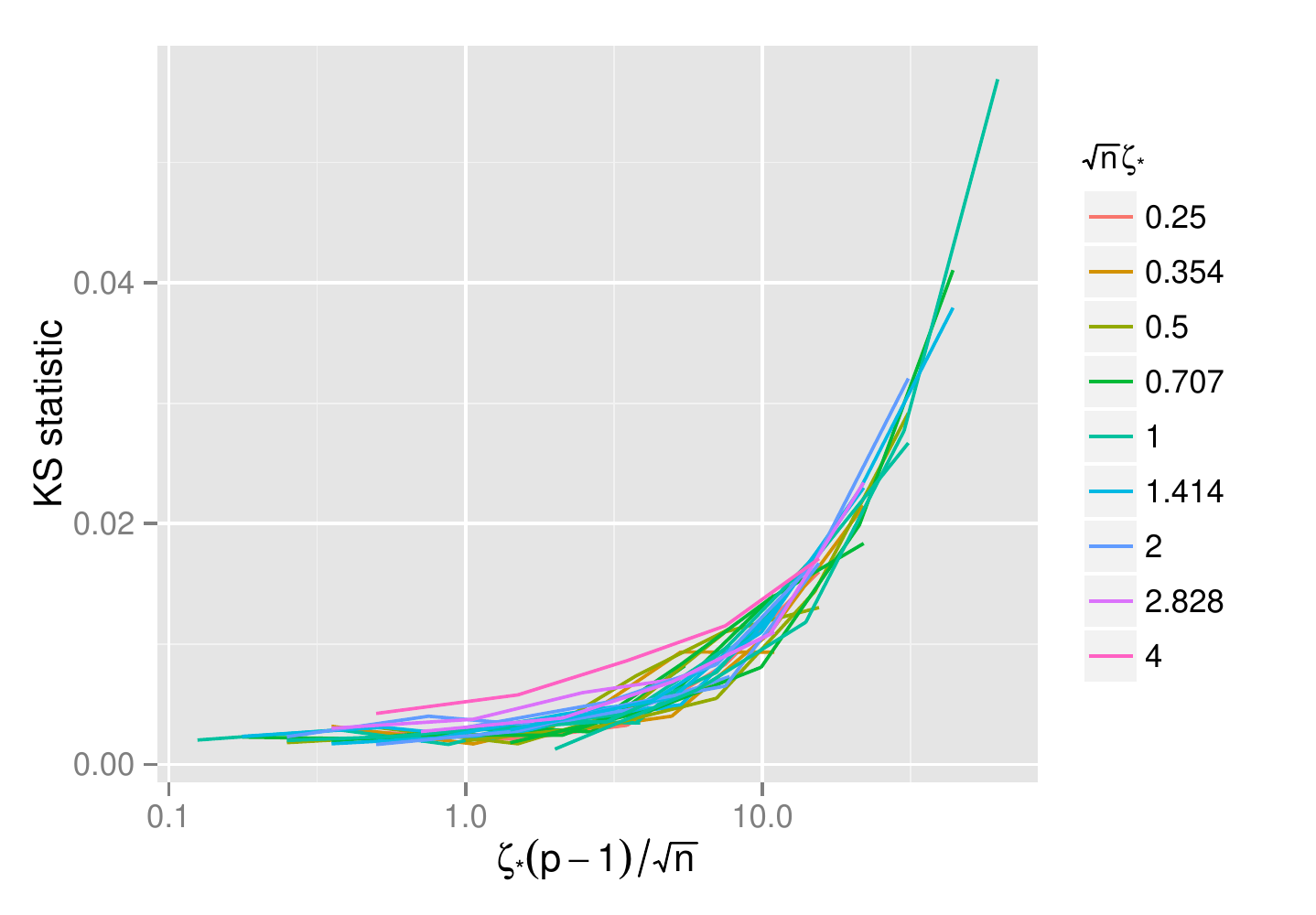} \caption[The \txtKS statistic for \apxref{qual_dist} over $10^{5}$ simulations of Gaussian returns is plotted versus $\psnropt \wrapParens{\nlatf-1}/\sqrt{\ssiz}$, with \psnropt in annualized terms, and \ssiz measured in years]{The \txtKS statistic for \apxref{qual_dist} over $10^{5}$ simulations of Gaussian returns is plotted versus $\psnropt \wrapParens{\nlatf-1}/\sqrt{\ssiz}$, with \psnropt in annualized terms, and \ssiz measured in years.  There is one line for each combination of $\ssiz$ and $\psnropt$. The line color corresponds to the `effect size', $\sqrt{\ssiz}\psnropt$, which is unitless.\label{fig:ksplots}}
\end{figure}

\end{knitrout}

Plots of the \txtKS statistic are given in \figref{ksplots}, and
\figref{ksheat}, which suggest that the quality of
\apxref{qual_dist} is a function of the
quantity $\psnropt \wrapParens{\nlatf-1}/\sqrt{\ssiz}$. 
As a rough guide, when
$\psnropt \wrapParens{\nlatf-1}/\sqrt{\ssiz} \le 5 \yrto{-1}$,
for daily observations, \apxref{qual_dist} is an acceptable 
approximation to the distribution of \txtQual of the sample \txtMP.


\begin{knitrout}\small
\definecolor{shadecolor}{rgb}{0.969, 0.969, 0.969}\color{fgcolor}\begin{figure}[]

\includegraphics[width=\maxwidth]{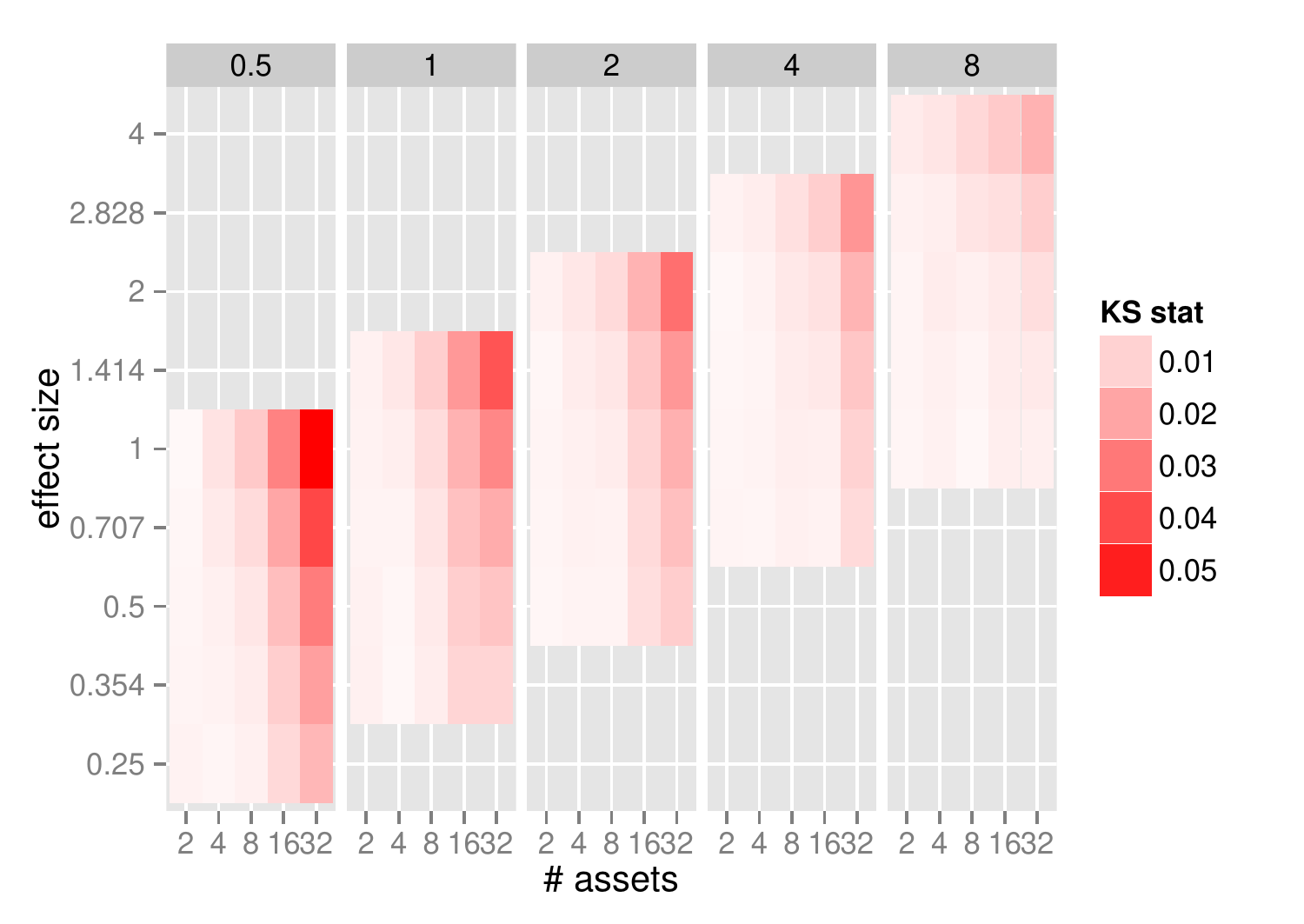} \caption[The \txtKS statistic for \apxref{qual_dist} over $10^{5}$ simulations of Gaussian returns is indicated, by color, versus \nlatf, and the `total effect size,' $\sqrt{\ssiz}\psnropt$, which is a unitless quantity]{The \txtKS statistic for \apxref{qual_dist} over $10^{5}$ simulations of Gaussian returns is indicated, by color, versus \nlatf, and the `total effect size,' $\sqrt{\ssiz}\psnropt$, which is a unitless quantity. Different facets are for different values of \ssiz (in years).\label{fig:ksheat}}
\end{figure}

\end{knitrout}



\section{Diversification}
\label{sec:diversification}

\theoremref{qual_bound} has implications for the diversification
benefit. 
Consider the case of $\nlatf=6, \ssiz=1012,
\psnropt = 1.25\yrtomhalf$ versus some
superset of this asset universe with 
$\nlatf=24, \ssiz=1012,
\psnropt = 1.6\yrtomhalf$. 
Since the optimum cannot
decrease over a larger feasible space, we observe that
the superset has a higher population \txtSNR, \psnropt, 
One should not, of course, increase
the investment universe without \emph{some} concomitant 
increase in \psnropt.
However, in this case the bound on expected \txtQual from
\theoremref{qual_bound} for the smaller asset universe
is $0.93\yrtomhalf$,
while for the superset it is
$0.89\yrtomhalf$. Diversification
has possibly caused a \emph{decrease} in expected \txtQual, even 
though the opportunity exists to increase \txtQual by a fair
amount.

By the `Fundamental Law of Asset Management,' one vaguely
expects \psnropt to increase as $\sqrt{\nlatf}$. \cite{grinold1999active}
If however, \psnropt scales at a rate slower than $\nlatf^{1/4}$,
then the derivative of the bound in \theoremref{qual_bound}
will be negative for sufficiently large 
$\nlatf$: adding assets to the universe causes a decrease in 
expected \txtQual.  To see why, note that
\fracc{\sqrt{\ssiz}\psnrsqopt}{\sqrt{\nlatf - 1 + \ssiz\psnrsqopt}}
has $\psnrsqopt$ in the numerator, and $\sqrt{\nlatf}$ in the denominator;
if \psnropt grows slower than $\nlatf^{1/4}$ the denominator will
outpace the numerator.

More formally, let \qbnd be the bound on \txtQual from 
\theoremref{qual_bound}:
\begin{equation*}
\qbnd\defeq\frac{\sqrt{\ssiz}\psnrsqopt}{\sqrt{\nlatf - 1 + \ssiz\psnrsqopt}}.
\end{equation*}
By taking the derivative of $\log\qbnd$ with respect to \nlatf, a little
calculus reveals that 
\begin{align}
\nonumber
\dbyd{\log\qbnd}{\nlatf} \ge 0 
&\Leftrightarrow
\frac{\psnropt}{2\ssiz\psnrsqopt + 4\wrapParens{\nlatf - 1}} \le
\dbyd{\psnropt}{\nlatf},&\\
&\Leftrightarrow
\frac{1}{2\ssiz\psnrsqopt + 4\wrapParens{\nlatf - 1}} \le
\dbyd{\log\psnropt}{\nlatf},&
\label{eqn:sufficient_growth}
\end{align}
The last inequality is implied by the inequality
$\oneby{4\wrapParens{\nlatf-1}} \le \dbyd{\log\psnropt}{\nlatf}$, 
with equality holding for $\psnropt = c \wrapParens{\nlatf-1}^{1/4}$.

\begin{knitrout}\small
\definecolor{shadecolor}{rgb}{0.969, 0.969, 0.969}\color{fgcolor}\begin{figure}[]

\includegraphics[width=\maxwidth]{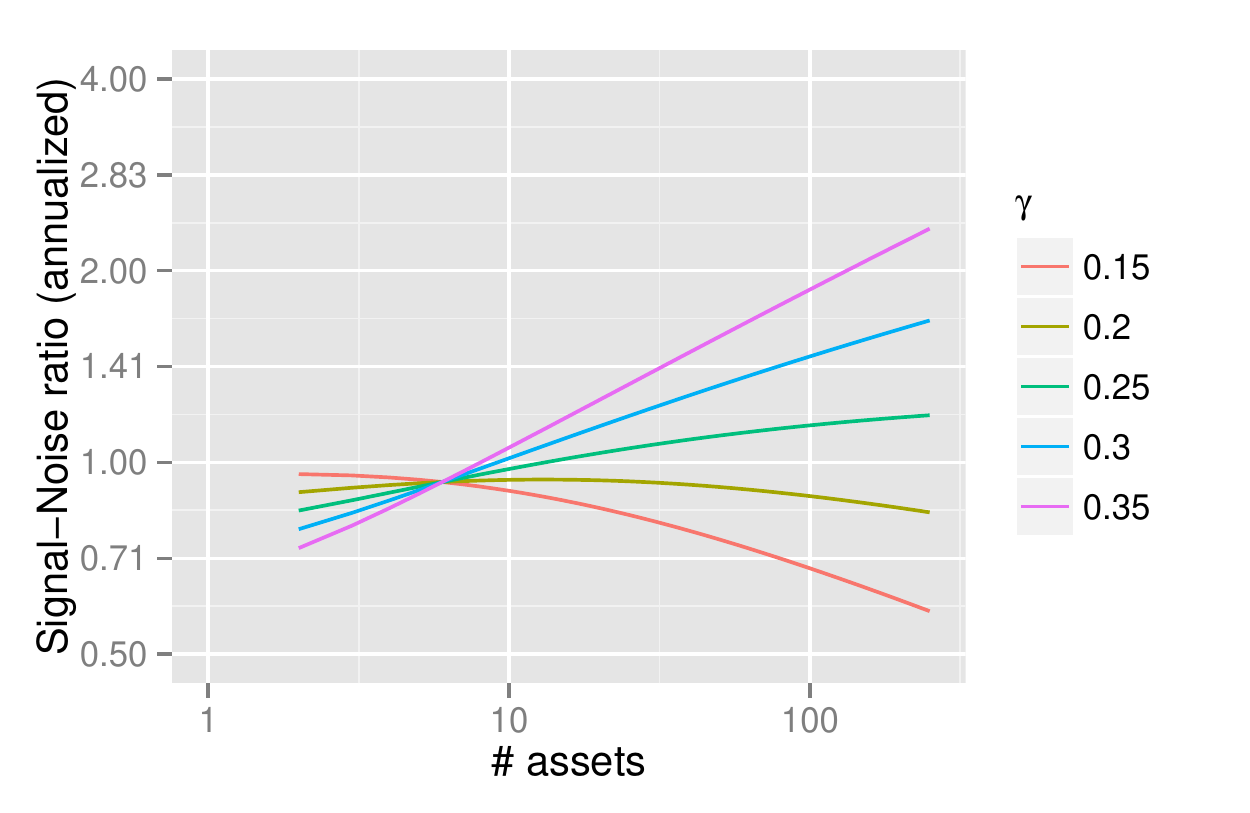} \caption[The upper bound of \theoremref{qual_bound} is plotted versus \nlatf for different scaling laws for \psnropt]{The upper bound of \theoremref{qual_bound} is plotted versus \nlatf for different scaling laws for \psnropt.  These scaling laws correspond to $\psnropt = \psnr[0]\nlatf^{\gamma}$, with $\gamma$ taking values between 0.15 and
0.35.  The constant terms, \psnr[0], are adjusted so that $\psnropt =1.25\yrtomhalf$ for $\nlatf = 6$ for all the lines. The bound uses $\ssiz=1012$, corresponding to 4 years of daily observations.\label{fig:grow_bound}}
\end{figure}

\end{knitrout}

The decreasing upper bound with respect to growing universe size
is illustrated in \figref{grow_bound}. Under the assumption
$\psnropt = \psnr[0]\nlatf^{\gamma}$, the upper bound
of \theoremref{qual_bound} is plotted versus \nlatf for 
different values of $\gamma$. 
The value of
$\psnr[0]$ is set so that 
$\psnropt=1.25\yrtomhalf$ when 
$\nlatf=6$.
For $\gamma < \oneby{4}$, one sees a local maximum in 
the upper bound as \nlatf increases, a behavior not seen for
$\gamma > \oneby{4}$, where the bound on \txtQual
grows with \nlatf.


\begin{knitrout}\small
\definecolor{shadecolor}{rgb}{0.969, 0.969, 0.969}\color{fgcolor}\begin{figure}[]

\includegraphics[width=\maxwidth]{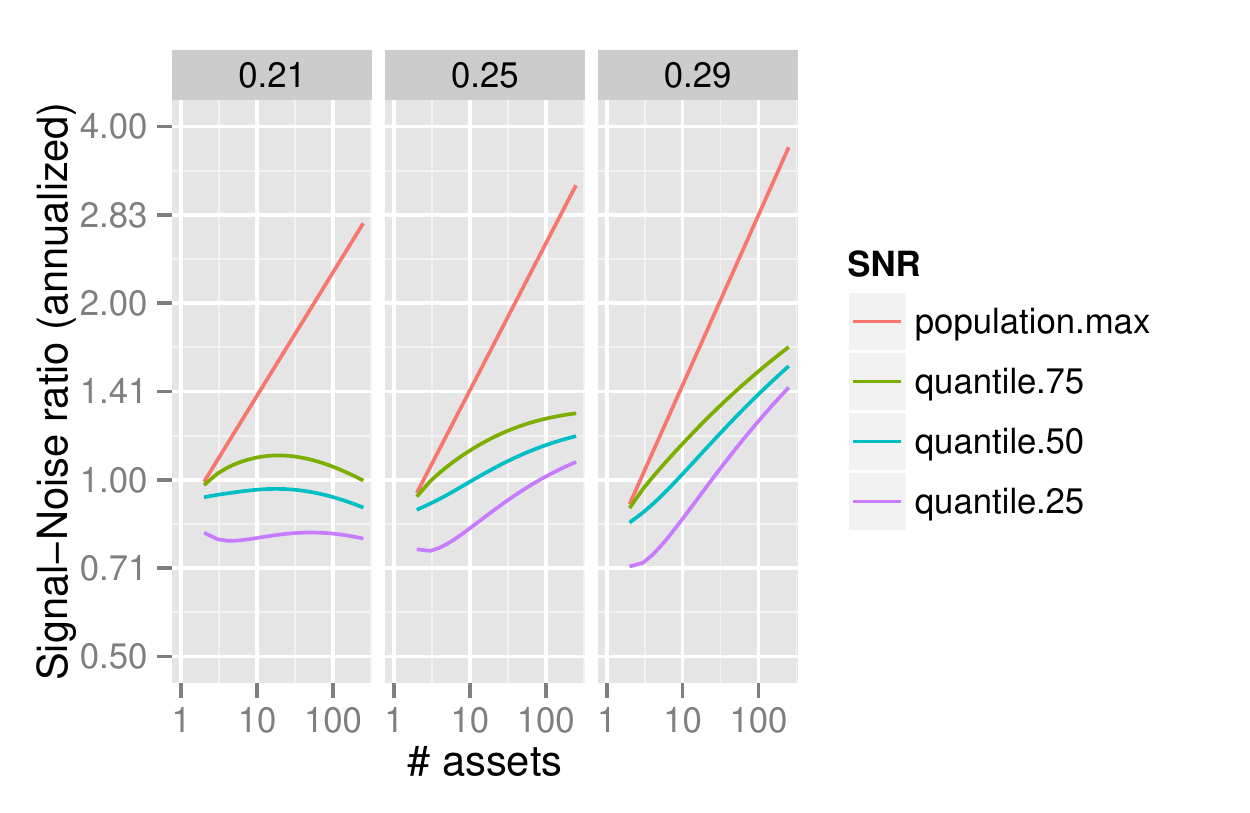} \caption[Some quantiles of the \txtQual of the \txtMP, under \apxref{qual_dist}, are plotted versus \nlatf for different scaling laws for \psnropt]{Some quantiles of the \txtQual of the \txtMP, under \apxref{qual_dist}, are plotted versus \nlatf for different scaling laws for \psnropt. The 3 panels represent different values of $\gamma$, \viz 0.21, 0.25, and 0.29. The bound uses $\ssiz=1012$, corresponding to 4 years of daily observations.\label{fig:grow_sqrt}}
\end{figure}

\end{knitrout}

This relationship between \txtQual and \nlatf for different
values of $\gamma$ appears not just in the upper bound of
\theoremref{qual_bound}, but apparently also for most quantiles
of the distribution given by \apxref{qual_dist}, as
illustrated in \figref{grow_sqrt}. Again assuming
$\psnropt = \psnr[0]\nlatf^{\gamma}$, lines of 
$\psnropt$ and the $0.25, 0.50,$ and $0.75$ quantiles of
$\psnropt\pql{\sportwfnc{\mreti}}$, under 
\apxref{qual_dist}, are plotted versus \nlatf.
The panels represent $\gamma$ values of
$0.21, 0.25,$ and 
$0.29$.
Again, the value of
$\psnr[0]$ is set so that 
$\psnropt=1.25\yrtomhalf$ when 
$\nlatf=6$.
For $\gamma < \oneby{4}$, one sees a local maximum in 
\txtQual as \nlatf increases, a behavior not seen for
$\gamma > \oneby{4}$, where quantiles of \txtQual grow with \nlatf. 
For the case of `slow growth' of \psnropt, the diversification 
benefit is not seen by the sample \txtMP, rather its practical
utility \emph{decreases} because the estimation error 
outpaces the growth of \psnropt. 


\subsection{Diversification under CAPM}

It is not clear how \psnropt `should' scale with \nlatf. It is
easy to construct a model under which \psnropt scales as
$\nlatf^{\half}$: assume all assets have independent returns with
the same \txtSNR. It is also easy to accidentally construct 
a model under which \psnropt ultimately scales 
as $\nlatf^{\epsilon}$ for small $\epsilon$, as done here.
Suppose the \kth{i} asset has expected return
$\alpha_i$, exposure $\beta_i$ to `the market', and volatility
$\sigma$. Assume the market return is zero mean with volatility
$\sigma_m$. Then the squared \txtSNR is
\begin{align}
\nonumber
\psnrsqopt 
&= 
\frac{\gram{\vect{\alpha}} +
\wrapParens{\frac{\sigma_m}{\sigma}}^2 \wrapBracks{%
\gram{\vect{\beta}}\gram{\vect{\alpha}} -
\wrapParens{\trAB{\vect{\alpha}}{\vect{\beta}}}^2}}{\sigma^2 +
\sigma_m^2\gram{\vect{\beta}}},&\\
&= 
\frac{\gram{\vect{\alpha}}}{\sigma^2}
\frac{\sigma^2 + \sigma_m^2 \gram{\vect{\beta}} \fsin[2]{\psi}}{%
\sigma^2 + \sigma_m^2\gram{\vect{\beta}}},&
\label{eqn:capm_zetas}
\end{align}
where $\psi$ is the angle between the vectors \vect{\alpha}
and \vect{\beta}.
Depending on how the sine of $\psi$ grows with universe size, 
one observes different scaling of \psnropt with respect to \nlatf. When
the assets all have the same alpha and beta, \ie
$\vect{\alpha} = \alpha\vone$ and $\vect{\beta} = \beta\vone$,
the sine is identically zero, and 
$\psnropt = \sqrt{\fraccp{\nlatf\alpha^2}{\sigma^2 + \nlatf\sigma_m^2\beta^2}}
< \alpha\beta^{-1}\sigma_m^{-1}$. Thus 
\psnropt asymptotically scales slower than $\nlatf^{\epsilon}$
for all $\epsilon > 0$.

On the other hand, when the sine is
one, \ie when \vect{\alpha} is orthogonal to \vect{\beta},
$\psnropt = \sqrt{\gram{\vect{\alpha}}} \sigma^{-1}$, which
grows however the assets are ordered, presumably on the
order of $\nlatf^{\half}$. Thus under a CAPM model, the
growth of \psnropt depends on the `alignment' of the
vectors \vect{\alpha} and \vect{\beta}.


\section{Generalizations}

\label{sec:generalizations}

\theoremref{qual_bound} is somewhat lacking because it ignores conditioning
information which may affect the distribution of future returns, and which
may inform the portfolio manager. Few active managers, it is presumed,
are holding the unconditional \txtMP based on in-sample data. What is sought
is a more general theorem that allows more elaborate models of returns,
and more elaborate, parametrized, trading schemes, with \psnropt redefined 
as the maximal portfolio \txtQual over the trading schemes, and \nlatf
redefined as the `degrees of freedom', perhaps the rank of some derivative
at the optimal parameter, say. Towards that goal, a few generalizations can
easily be made.

\subsection{Conditional portfolio \txtQual}
\label{subsec:cond_portfolio_qual}

The model of stationary mean returns is generalized by one where
the expected return of the assets is linear in some state
variables, or `features', \vfact[i], observed prior to the investment
decision.  \cite{pav2013markowitz,connor1997,herold2004TAA} That
is, one observes the \nfac-vector $\vfact[i]$ at some time prior to 
when the investment decision is required to capture \vreti[i+1]. 
The general model is now
\begin{align}
\label{eqn:cond_model_IV}
\Econd{\vreti[i+1]}{\vfact[i]} &=
\pRegco \vfact[i], &
\Varcond{\vreti[i+1]}{\vfact[i]} &=
\pvsig,
\end{align}
where \pRegco is some \bby{\nlatf}{\nfac} matrix. 

Here we bound the \txtQual of portfolios which are
linear in the features \vfact[i]. That is, the portfolio
manager allocates their assets proportional to
$\sportW\vfact[i]$ for some matrix $\sportW$.

Using the law of iterated expectations, the unconditional
expected value of the returns of the portfolio is
\begin{equation*}
\E{\Econd{\trAB{\wrapParens{\sportW\vfact[i]}}{\vreti[i+1]}}{\vfact[i]}}
= \trace{\tr{\sportW}\pRegco\E{\ogram{\vfact[i]}}}
= \trace{\tr{\sportW}\pRegco\pfacsig},
\end{equation*}
by definition of \pfacsig as the second moment of \vfact[i].

Unfortunately the unconditional variance will, in general,
involve a term quadratic in the expectation. However, it
can easily be shown that the unconditional \emph{expected} 
variance of the portfolio's returns is
\begin{equation*}
\E{\qform{\pvsig}{\wrapParens{\sportW\vfact[i]}}} =
\trace{\qform{\pvsig}{\sportW}\pfacsig}.
\end{equation*}
We can then redefine\footnote{If an analysis of the conditional expected
return divided by risk is required, it is possible one could define
\pQl{\cdot} as the expected return divided by square root of the unconditional
second moment. The \txtQual would then be $\ftan{\farcsin{\pQl{\cdot}}}$.
One could possibly find a \txtCR bound on the expected value of this \pQl{\cdot}.
This `Pillai-Bartlett' form of \pQl{\cdot} is likely unrequired for low frequency
settings.} the \txtQual of the portfolio as the
unconditional mean divided by the unconditional expected risk:
\begin{equation}
\pQl{\sportW}\defeq\frac{\trace{\tr{\sportW}\pRegco\pfacsig}}{%
\sqrt{\trace{\qform{\pvsig}{\sportW}\pfacsig}}}.
\end{equation}
When \vfact[i] is a deterministic scalar constant, 
this coincides with the `usual' definition of \txtQual as being 
like a \txtSR. However, except possibly for an intercept term, one
expects \vfact[i] to be random, or at least out of the control of
the portfolio manager.

Once again, a risk transform can be injected to express
portfolio optimization as an estimation problem on a sphere:
\begin{equation}
\begin{split}
\pQl{\sportW}
&=\frac{\trace{\trAB{%
\wrapParens{\trchol{\pvsig}\sportW\chol{\pfacsig}}}{%
\wrapParens{\ichol{\pvsig}\pRegco\chol{\pfacsig}}}{%
}}}{%
\sqrt{\trace{
\gram{\wrapParens{\trchol{\pvsig}\sportW\chol{\pfacsig}}}}}},\\
&=
\frac{\trAB{\fvec{\trchol{\pvsig}\sportW\chol{\pfacsig}}}{%
\fvec{\ichol{\pvsig}\pRegco\chol{\pfacsig}}}}{%
\sqrt{\gram{\fvec{\trchol{\pvsig}\sportW\chol{\pfacsig}}}}}.
\end{split}
\end{equation}
This function is maximized by taking
\begin{equation}
\sportW = \pportWopt \defeq \minv{\pvsig}{\pRegco},
\end{equation}
which has \txtQual
\begin{equation}
\psnropt\defeq\pQl{\pportWopt} =
\sqrt{\trace{\qiform{\pvsig}{\pRegco}\pfacsig}}.
\end{equation}
The square of this quantity, \psnrsqopt, 
is the `population analogue' of the Hotelling-Lawley trace. \cite{Rencher2002,Muller1984143}

Again we can write
\begin{equation*}
\frac{\pQl{\sportW}}{\psnropt} = 
\trAB{\fnorm{\fvec{\trchol{\pvsig}\sportW\chol{\pfacsig}}}}{%
\fnorm{\fvec{\ichol{\pvsig}\pRegco\chol{\pfacsig}}}}.
\end{equation*}
Thus finding a `good' \sportW becomes an estimation problem on the
sphere \sphere{\nfac\nlatf - 1}. An analogue
to \theoremref{qual_bound} can be proved with $\nfac\nlatf$ replacing
$\nlatf$, by assuming a particular form to the likelihood. We must
generalize the assumption of Directional Independence, after which
the theorem proceeds easily.

\begin{assumption}[Conditional Directional Independence]
Assume that
\begin{equation}
\label{eqn:sane_estimator_cond}
\E{\fnorm{\fvec{\trchol{\pvsig}\sportWfnc{\mreti,\mfact}}}} = 
\cfnc{\psnrsqopt} \fnorm{\fvec{\trchol{\pvsig}\pportWopt\chol{\pfacsig}}}
+ \Bterm,
\end{equation}
where \Bterm is the bias term, orthogonal to 
\fnorm{\fvec{\trchol{\pvsig}\pportWopt\chol{\pfacsig}}}.
\end{assumption}

\begin{theorem}
\label{theorem:qual_bound_two}
Let one element of \vfact[i] be a deterministic $1$. Suppose the
vector of the remaining $\nfac-1$ elements of \vfact[i] stacked 
on top of \vreti[i+1] are multivariate Gaussian.  Let \mreti, 
\mfact be \bby{\ssiz}{\nlatf} and \bby{\ssiz}{\nfac} matrices
of \iid observations of the features and returns.
Let \sportWfnc{\mreti,\mfact} be an estimator 
satisfying the assumptions of
Conditional Directional Independence and Residual Independence. Then
\begin{equation}
\E{\pQl{\sportWfnc{\mreti,\mfact}}} 
\le \frac{\sqrt{\ssiz}\psnrsqopt}{\sqrt{\nfac\nlatf - 1 + \ssiz\psnrsqopt}}.
\end{equation}
\end{theorem}
\begin{proof}

We can proceed as in \secref{portfolio_qual}. Let \mreti be the 
\bby{\ssiz}{\nlatf} matrix of portfolio returns, and let 
\mfact be the corresponding \bby{\ssiz}{\nfac} matrix of features.
View the portfolio coefficient \sportW as an estimator, a function of
the random data, \ie \sportWfnc{\mreti,\mfact}.
Define
\begin{equation}
\prskMtx\defeq{\ichol{\pvsig}\pRegco\chol{\pfacsig}}.
\end{equation}
Then 
\begin{equation*}
\psnrsqopt = \trace{\gramprskMtx}.
\end{equation*}

We get, analogously to \eqnref{crb_one},
\begin{equation}
\oneby{\ssiz}\trace{\qoform{\iFishI[\fvec{\prskMtx}]}{\Drv}}
\le 1 - \csqfnc{\trace{\gramprskMtx}},
\end{equation}
where
\begin{equation}
\Drv\defeq
{\dbyd{\cfnc{\trace{\gramprskMtx}}\frac{\prskMtx}{\sqrt{\trace{\gramprskMtx}}}}{\fvec{\prskMtx}}}.
\end{equation}

Without loss of generality, we assume it is the first element
of \vfact[i] that is a deterministic $1$. Then, the log likelihood of 
the vector of $\vfact[i]$ stacked on top of $\vreti[i+1]$
is: \cite{pav2013markowitz}
\begin{equation}
\label{eqn:cond_llik_one}
\log \FOOlik{}{\pvsm}{\twobyone{\vfact[i]}{\vreti[i+1]}} = 
  c_{\nfac+\nlatf} 
- \half \logdet{\pvsm} 
- \half \trace{\minv{\pvsm}\ogram{\twobyone{\vfact[i]}{\vreti[i+1]}}},
\end{equation}
where \pvsm is the second moment matrix:
\begin{equation}
\pvsm \defeq \E{\ogram{\twobyone{\vfact[i]}{\vreti[i+1]}}}
= \twobytwo{\pfacsig}{\pfacsig\tr{\pRegco}}{\pRegco\pfacsig}{\pvsig +
\qoform{\pfacsig}{\pRegco}}.
\end{equation}
The inverse of \pvsm has the following, somewhat surprising, form \cite{pav2013markowitz}:
\begin{equation}
\minv{\pvsm} 
= \twobytwo{\minv{\pfacsig} +
\qiform{\pvsig}{\pRegco}}{-\tr{\pRegco}\minv{\pvsig}}{-\minv{\pvsig}\pRegco}{\minv{\pvsig}}.
\end{equation}
A square root of this matrix (a Cholesky factor, up to permutation)
is:
\begin{equation}
\begin{split}
\minv{\pvsm} &=
\ogram{%
\twobytwo{\ichol{\pfacsig}}{-\tr{\pRegco}\ichol{\pvsig}}{\mzero}{\ichol{\pvsig}}
},\\
&= \twobytwo{\ichol{\pfacsig}}{\mzero}{\mzero}{\eye}
\ogram{%
\twobytwo{\eye}{-\prskMtxU{\trsym}}{\mzero}{\ichol{\pvsig}}}
\twobytwo{\trichol{\pfacsig}}{\mzero}{\mzero}{\eye}.
\end{split}
\end{equation}

By the block determinant formula, 
\begin{equation}
\det{\pvsm} 
= \det{\pfacsig}\det{\pvsig + \qoform{\pfacsig}{\pRegco} -
\pRegco\pfacsig\minv{\pfacsig}\pfacsig\tr{\pRegco}}
= \det{\pfacsig}\det{\pvsig}.
\end{equation}
Thus, conditional on \pfacsig and \pvsig, the negative log likelihood
takes the form:
\begin{multline}
\label{eqn:cond_llik_two}
- \log \FOOlik{}{\prskMtx, \pfacsig, \pvsig}{\twobyone{\vfact[i]}{\vreti[i+1]}} = 
- c_{\nfac+\nlatf}
+ \half \logdet{\pfacsig} 
+ \half \logdet{\pvsig}\\
+ \half \trace{
\ogram{%
\twobytwo{\eye}{-\prskMtxU{\trsym}}{\mzero}{\ichol{\pvsig}}}
\ogram{\twobyone{\trichol{\pfacsig}\vfact[i]}{\vreti[i+1]}}}.
\end{multline}
Sweeping the nuisance parameter terms into the constant, as well
as terms in the trace which are not quadratic in \prskMtx, we
have 
\begin{align}
\label{eqn:cond_llik_three}
- \log \FOOlik{}{\prskMtx, \pfacsig, \pvsig}{\twobyone{\vfact[i]}{\vreti[i+1]}}
&= - c'
+ \half \trace{\gramprskMtx 
\ogram{\wrapParens{\trichol{\pfacsig}\vfact[i]}}},&\\
&= - c'
+ \half
\tr{\fvec{\prskMtx}}\fvec{\prskMtx\trichol{\pfacsig}\ogram{\vfact[i]}\ichol{\pfacsig}}.&\\
&= - c'
+ \half
\tr{\fvec{\prskMtx}}\wrapParens{%
\wrapBracks{\trichol{\pfacsig}\ogram{\vfact[i]}\ichol{\pfacsig}}
 \kron \eye}\fvec{\prskMtx}.&
\end{align}
The Fisher Information, then, is
\begin{equation}
\FishI[\fvec{\prskMtx}] =
\E{\wrapParens{%
\wrapBracks{\trichol{\pfacsig}\ogram{\vfact[i]}\ichol{\pfacsig}}
 \kron \eye}} = \eye[\nfac\nlatf].
\end{equation}

The remainder of the proof proceeds exactly as in 
\secref{portfolio_qual}. 
\end{proof}





\subsection{Subspace constraints}

Consider, now, the case of conditional expectation, as presented in
\subsecref{cond_portfolio_qual}, but where the portfolio is 
constrained to be in some lower dimensional subspace. That is, by design, 
\begin{equation}
\zerJc \sportWfnc{\mreti,\mfact} = \vzero,
\end{equation}
where $\zerJc$ is a \bby{\wrapParens{\nlatf - \nlatfzer}}{\nlatf}
matrix of rank $\nlatf - \nlatfzer$, that is chosen indpendently
of the observations of \mreti and \mfact.
Let the rows of \zerJ span the null space of the rows of
\zerJc; that is, $\zerJc \tr{\zerJ} = \mzero$, and $\ogram{\zerJ} = \eye$.

We can simply use the results of \subsecref{cond_portfolio_qual}, but
replacing the assets with the \nlatfzer assets spanned by the rows of
\zerJ. That is, we can replace the \vreti[i+1] with $\zerJ\vreti[i+1]$,
and replace \sportWfnc{\mreti,\mfact} with
$\tr{\zerJ}\minv{\wrapParens{\ogram{\zerJ}}}\sportWfnc{\mreti,\mfact}$
to arrive at the following analogue of \theoremref{qual_bound_two}:

\begin{theorem}
\label{theorem:qual_bound_three}
Let one element of \vfact[i] be a deterministic $1$. Suppose the
vector of the remaining $\nfac-1$ elements of \vfact[i] stacked 
on top of \vreti[i+1] are multivariate Gaussian.  Let \mreti, 
\mfact be \bby{\ssiz}{\nlatf} and \bby{\ssiz}{\nfac} matrices
of \iid observations of the features and returns.
Let \sportWfnc{\mreti,\mfact} be an estimator 
satisfying the assumptions of
directional independence and residual independence, with the constraint
\begin{equation}
\zerJc \sportWfnc{\mreti,\mfact} = \vzero,
\end{equation}
for \bby{\wrapParens{\nlatf - \nlatfzer}}{\nlatf} matrix \zerJc, which
is chosen independently of the observed \mreti and \mfact. 
Let the rows of \zerJ span the null space of the rows of
\zerJc. 

Then
\begin{equation}
\E{\pQl{\sportWfnc{\mreti,\mfact}}} 
\le \frac{\sqrt{\ssiz}\psnrsqoptG{\zerJ}}{\sqrt{\nfac\nlatfzer - 1 +
\ssiz\psnrsqoptG{\zerJ}}},
\end{equation}
where 
\begin{equation*}
\psnrsqoptG{\zerJ}\defeq 
\trace{\qform{\wrapProj{\pvsig}{\zerJ}}{\pRegco}\pfacsig}.
\end{equation*}
\end{theorem}

\subsection{Hedging constraints}

Consider, now, the case where one seeks a portfolio whose returns are
independent, in the probabilistic sense, of the returns of some 
traded instruments in the investment universe. 
Independence is a difficult property to
check or enforce; however, independence implies zero covariation, which
can be easily formulated and checked. 

Since the portfolio estimator may not deliver a perfectly
hedged portolio due to misestimation of the covariance matrix, we
will, with perfect knowledge of \pvsig, consider the \txtQual of the 
hedged part of the portfolio. The hedged
part is defined in terms of a risk projection. 
If \sportw[1] is a feasible portfolio based on the sample,
then the hedged version of this portfolio is the solution to the
optimization problem
\begin{equation}
\min_{\sportw : \hejG\pvsig\sportw = \vzero} \VAR{\trAB{\wrapParens{\sportw -
\sportw[1]}}{\vreti[i+1]}},
\end{equation}
where $\hejG$ is a \bby{\nlatfhej}{\nlatf} matrix of 
rank \nlatfhej, the rows of which we wish to `hedge out.'

Using the Lagrange multiplier technique, this can easily be found to
be solved by 
\begin{equation}
\sportw = \sportw[1] - \wrapProj{\pvsig}{\hejG}\pvsig\sportw[1].
\end{equation}
Thus we will consider the \txtQual of the portfolio estimator
\begin{equation*}
\wrapParens{\eye[\nlatf] - \wrapProj{\pvsig}{\hejG}\pvsig}\sportWfnc{\mreti,\mfact}.
\end{equation*}
Note, however, that the row rank of
$\wrapParens{\eye[\nlatf] - \wrapProj{\pvsig}{\hejG}\pvsig}$ 
is $\nlatf - \nlatfhej$.  Thus hedging is an instance of a subspace
constraint and we can apply \theoremref{qual_bound_three} outright.

\begin{theorem}
\label{theorem:qual_bound_four}
Let one element of \vfact[i] be a deterministic $1$. Suppose the
vector of the remaining $\nfac-1$ elements of \vfact[i] stacked 
on top of \vreti[i+1] are multivariate Gaussian.  Let \mreti, 
\mfact be \bby{\ssiz}{\nlatf} and \bby{\ssiz}{\nfac} matrices
of \iid observations of the features and returns.
Let \sportWfnc{\mreti,\mfact} be an estimator 
satisfying the assumptions of
directional independence and residual independence. 
Let \bby{\nlatfhej}{\nlatf} matrix \hejG be chosen
independently of \mreti and \mfact. 

Define 
\begin{equation}
\Delpsnrsqopt{\eye,\hejG}
\defeq
\trace{\qiform{\svsig}{\pRegco}\pfacsig} - 
\trace{\qform{\wrapProj{\pvsig}{\hejG}}{\pRegco}\pfacsig}.
\end{equation}

Then
\begin{equation}
\E{\pQl{%
\wrapBracks{\eye[\nlatf] - \wrapProj{\pvsig}{\hejG}\pvsig}\sportWfnc{\mreti,\mfact}}} 
\le \frac{\sqrt{\ssiz}\Delpsnrsqopt{\eye,\hejG}}{\sqrt{\nfac\wrapParens{\nlatf -
\nlatfhej} - 1 + \ssiz\Delpsnrsqopt{\eye,\hejG}}}.
\end{equation}
\end{theorem}



\section{Examples}

\subsection{The equal weight puzzle}

\theoremref{qual_bound} can help us make sense of puzzling findings
in the literature. For example, in the ``$1/N$'' paper, DeMiguel \etal
find that the equal-weighting portfolio outperforms, in terms of out-of-sample
\txtSR (and other measures), the \txtMP and numerous other portfolio
estimators.  \cite{demiguel2009optimal}
This finding is supported on a number of real world data
sets, and a few synthetic ones. One data set used was the returns of
the 10 industry portfolios and the US equity market portfolio, computed
by Ken French. 

The monthly returns, from 1927-01-01 to 
2013-08-01, for these 11 
assets were downloaded
from \emph{Quandl}.  \cite{Quandl}
The \txtSR of the equal weighted portfolio on the 
assets, over the 1040 months, is around 
$0.65\yrtomhalf$. The \txtSR of the sample \txtMP over
the 11 assets over the same period is around
$0.99\yrtomhalf$.  \cite{SharpeR-Manual}
Now consider a portfolio estimator
given 5 years of observations, as in 
DeMiguel \etal \cite{demiguel2009optimal}, assuming 
$\psnropt=0.99\yrtomhalf$. The bound on expected
value of \pql{\sportwfnc{\mreti}} from \theoremref{qual_bound}
is only $0.57\yrtomhalf$.  Under \apxref{qual_dist}, 
the probability that \pql{\sportwfnc{\mreti}} exceeds 
$0.65\yrtomhalf$ in this case is only
$0.33$. It is not surprising that DeMiguel \etal drew
the conclusions they did, nor that they would be refuted
by looking at a longer sample, as by Kritzman \etal \cite{defoopt2010}

One could also use \theoremref{qual_bound_four} here. However,
the upper bound of that theorem is non-negative, and zero only
if the quantity \Delpsnrsqopt{\eye,\hejG} is zero.  This is a
statement regarding unknown population parameters, but we can
perform inference on this quantity. For example, based on the
1040 months of data on these 
11, the 95\% confidence interval on
\Delpsnrsqopt{\eye,\hejG}, where \hejG is the \bby{1}{11}
matrix of all ones, is
$\asrowvec{0.18, 0.79}\yrto{-1}$, under the assumption of
Gaussian returns.  \cite{SharpeR-Manual}






\subsection{Empirical diversification in the S\&P 100}

To check how \psnropt \emph{might} scale with \nlatf, the weekly
log returns of the adjusted close prices of the stocks in the 
S\&P 100 Index, as of March 21, 2014, were downloaded 
from \emph{Quandl}.  \cite{Quandl} 
Adjustments for splits and
dividends were made in some unspecified way by the upstream source
of the data, Yahoo Finance. Stocks without a full 5 years of history
were discarded, leaving 96 stocks. 
Note that selection based on membership in the index at the end of the 
period adds no small amount of selection bias, which we shall ignore 
here.

Based on the weekly returns from 2009-03-27 to 2014-04-04,
estimates of \psnropt were computed, using the `KRS' estimator.
\cite{kubokawa1993estimation,SharpeR-Manual} 
This was performed on the first \nlatf assets, with \nlatf ranging from 
$1$ to $96$.  The estimate of \psnropt versus \nlatf
is plotted in \figref{sp100_grow}, with assets added in alphabetical order.
Because Apple appears at the beginning of this list, it appears that
\psnropt starts reasonably large, but then actually \emph{decreases}
when adding assets. This is an artifact of the estimator, since the true
\psnropt can only increase when adding assets. 

\begin{knitrout}\small
\definecolor{shadecolor}{rgb}{0.969, 0.969, 0.969}\color{fgcolor}\begin{figure}[]

\includegraphics[width=\maxwidth]{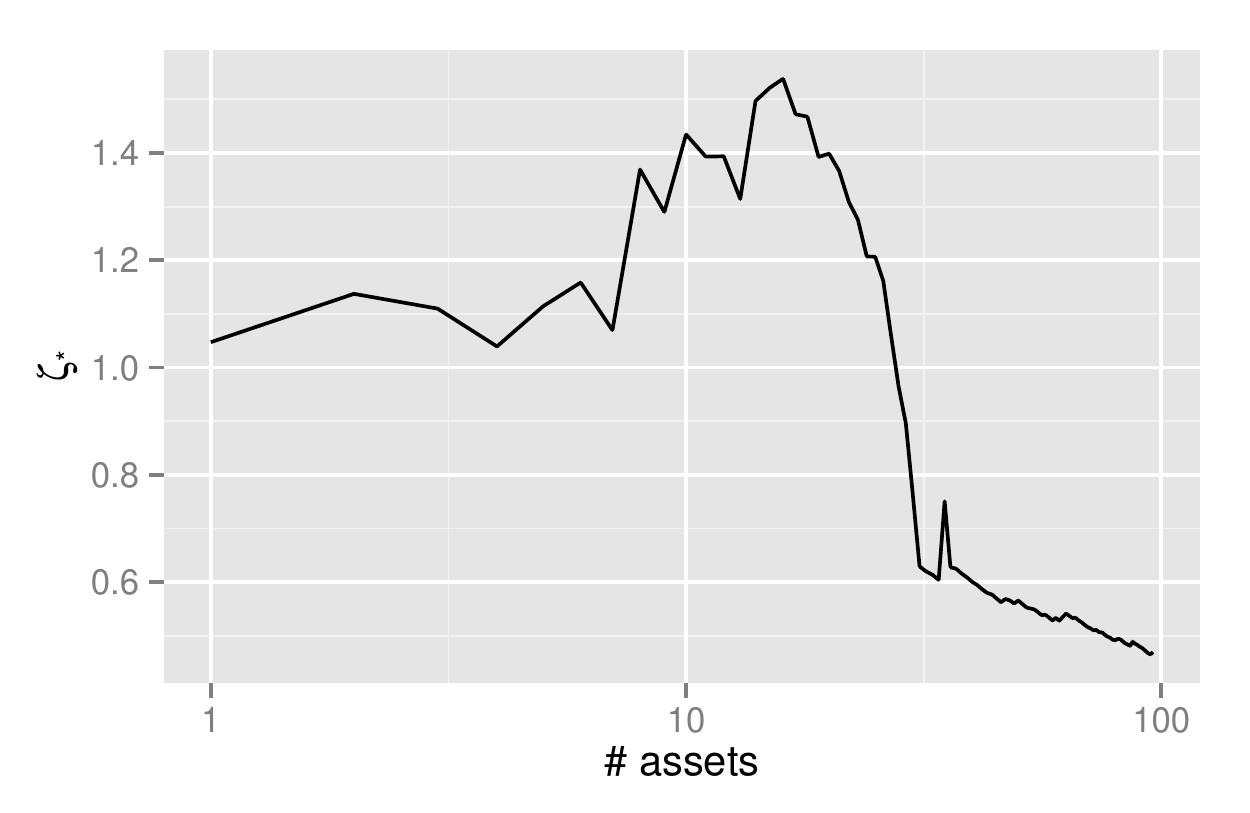} \caption[Growth of estimated \psnropt versus \nlatf for the S\&P 100 Index names, in alphabetical order, showing the `Apple effect]{Growth of estimated \psnropt versus \nlatf for the S\&P 100 Index names, in alphabetical order, showing the `Apple effect.'\label{fig:sp100_grow}}
\end{figure}

\end{knitrout}





Since the ordering of assets here is arbitrary, the experiment was repeated
1000 times, with the stocks randomly permuted, and 
\psnropt estimated as a function of \nlatf. Boxplots, over the
1000 simulations, of the KRS statistic versus
\nlatf are given in \figref{sp100_box}. There is effectively no 
diversification benefit observed here beyond the mean effect, which is
equivalent to holding an equal weight portfolio. Given the 
conditions under which \txtQual grows with \nlatf outlined in 
\secref{diversification}, one expects poor performance of 
directionally independent portfolio estimators
over even a small subset of the S\&P 100.

\begin{knitrout}\small
\definecolor{shadecolor}{rgb}{0.969, 0.969, 0.969}\color{fgcolor}\begin{figure}[]

\includegraphics[width=\maxwidth]{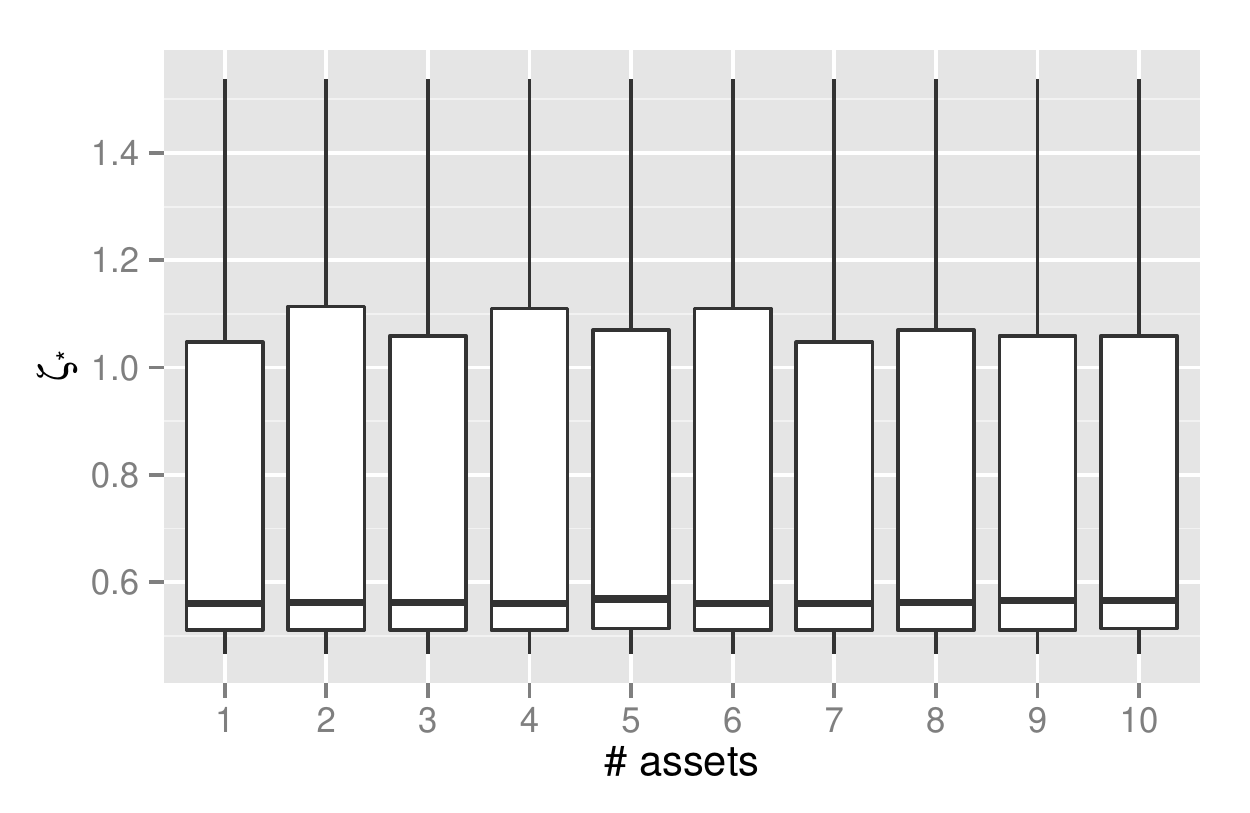} \caption[Growth of estimated \psnropt versus \nlatf for the S\&P 100 Index names is shown over 1000 permutations of the stocks]{Growth of estimated \psnropt versus \nlatf for the S\&P 100 Index names is shown over 1000 permutations of the stocks. There is effectively \emph{no} diversification benefit here beyond an equal weight portfolio.\label{fig:sp100_box}}
\end{figure}

\end{knitrout}

\section{Discussion}

Care should be taken in the interpretation of \theoremref{qual_bound},
or its generalizations from \secref{generalizations}.
It does not claim that the sample \txtMP is somehow `optimal,' nor
does it make comparative claims about different portfolio estimators
when presented with the same data.
The theorem does not imply that somehow `overfitting' to the observed
data can be mitigated by selecting a less desireable portfolio. 
It does not claim that sample estimates of the \txtQual of a portfolio
are useless.  It is trivially
the case, for example, that if $\pql{\sportw[1]} > \pql{\sportw[2]}$, then,
with probability greater than half, 
$\trAB{\sportw[1]}{\svmu} / \sqrt{\qform{\svsig}{\sportw[1]}} > 
\trAB{\sportw[2]}{\svmu} / \sqrt{\qform{\svsig}{\sportw[2]}}$, where
the probability is over draws of \svmu and \svsig.
The theorem does not claim that the expected
\txtQual of a portfolio estimator is negative. (Indeed, it can not,
since the portfolio estimator which generates a random portfolio, 
ignoring the data, has zero expected \txtQual). 
The theorem makes no claims (\eg providing a Bayesian posterior) about 
any particular portfolio based on a single sample of the data: it is
a statement about the expectation of the \emph{estimator} under replication 
of draws of the sample.

One should recognize, moreover, there are situations where the assumptions
of the theorem are violated. For example, in some cases a prior 
bias for positive expected returns, \ie $\pvmu \ge 0$, is warranted,
and thus a portfolio estimator with a long bias is chosen.
This can happen when the underlying assets are equities, and
the eligible universe is based on some minimum longevity, as this 
introduces a `good' survivorship bias: 
companies with negative expected return should founder and perish, 
leaving behind those with more positive $\pmu$. Effectively this
acts to boost \ssiz somewhat, although the effect is likely small.

There are other reasonable portfolio estimators which violate the 
assumption of Directional Independence. For example, an estimator
which performs some dimensionality reduction based on the observed
data, \mreti and \mfact will not be covered by 
\theoremref{qual_bound_three} since the subspace is chosen based
on the sample. However, it might not be covered by 
\theoremref{qual_bound_two} because the expected \txtQual might
depend on how \pRegco aligns with the leading eigenvectors of
\pvsig, say.

\subsection{Future work}

These findings perhaps raise more questions than they answer:
\begin{compactenum}
\item Foremost, the bounds of 
\theoremref{qual_bound} and \theoremref{qual_bound_two}
depend on the unknown
quantity, \psnrsqopt. How can we perform inference, 
Frequentist or Bayesian, on \pql{\sportwopt}, where \sportwopt is the \txtMP, 
given the observed information (\viz \svmu and \svsig)? This
is a problem of enormous practical concern to hundreds of
quantitative portfolio managers.

Contrast inference on the portfolio \txtQual with inference
on the population \txtSNR: under Gaussian returns, the
distribution of \ssrsqopt in terms of \ssiz, \nlatf and \psnrsqopt
is known. \cite[Theorem 5.2.2]{anderson2003introduction}
Thus, for example, the quantity
$\wrapParens{1 - \fracc{\nlatf}{\ssiz}}\ssrsqopt - \fracc{\nlatf}{\ssiz}$
is an unbiased estimator for \psnrsqopt, \etc 
Performing inverence on \pql{\sportwopt} is tricky because \psnrsqopt
is unknown and the error $\sportwopt - \pportwopt$ is likely
not independent of the error in the estimate \ssropt.

It may be the case, however, that inference on the portfolio
\txtQual qualifies as an `impossible' estimation-after-selection
problem. \cite{leeb2006}
\item While \theoremref{qual_bound} requires Gaussian returns, 
one expects that the result holds for returns distributions whose 
likelihood is ``more concave'' than the Gaussian at the MLE.
Exact conditions for this to hold should be established.
\item \theoremref{qual_bound_two} applies to the case of 
trading strategies where the portfolio is linear in the 
observable features, \vfact[i]. 
Can it be used as an approximate bound for trading
strategies which are nonlinear, complex functions of the features?
\item What can be said about scaling of \psnropt with respect to
\nlatf for different models of market returns? Can one establish
sane sufficient conditions for which \psnropt grows slower than
$\nlatf^{\oneby{4}}$? What is the analogue of \eqnref{capm_zetas}
for a multi-factor model of returns?
\item Can we find a lower bound, or a non-trivial upper bound on the 
variance of $\pql{\sportwfnc{\mreti}}$? Together these could be used
to give guarantees about the quantiles of \pql{\sportwfnc{\mreti}}.
A lower bound on the variance can likely be had via a result of
Kakarala and Watson. \cite{ANZS:ANZS253} Together with Cantelli's
Inequality, these would give rough (perhaps useless) upper bounds
on the 
\kth{\qlev} quantile of portfolio \txtQual, for $\half < \qlev < 1$.
\item How tight is the bound of \theoremref{qual_bound}, and can 
it be much improved by directly analyzing the differential
inequality of \eqnref{crb_three}, rather than discarding
the derivative term? Or perhaps the bound can be improved by using
an `intrinsic' \txtCR bound.  \cite{DBLP:conf/icassp/XavierB05}
\item How good is \apxref{qual_dist}? Can we find the expected
value of the distribution in \apxref{qual_dist}, and what is the
gap between it and the bound of \theoremref{qual_bound}?
Can we find the \emph{exact} distribution of \txtQual of the 
sample \txtMP under Gaussian returns, 
perhaps leveraging the work of Bodnar and Okhrin,
or of Britton-Jones.  \cite{SJOS:SJOS729, BrittenJones1999}
\item Can the assumption of Directional Independence be weakened?
Can the \theoremref{qual_bound_two} be generalized to deal with omitted 
variable bias in \vfact[i]?  
\item The analysis of \txtQual ignores the `risk-free' or `disastrous'
rate of return, and all trading costs. Can the expected bounds
be generalized to include these costs?
\end{compactenum}

\bibliographystyle{plainnat}
\bibliography{SharpeR,rauto}

\appendix




\end{document}